\documentclass[conference]{IEEEtran}
\usepackage{soul,xcolor}
\usepackage{bm}
\include{graphics}
\include{amsmath}
\usepackage[normalem]{ulem}
\usepackage{multirow,enumitem}
\usepackage{algorithm}
\usepackage{algcompatible}
\usepackage{verbatim}
\usepackage[latin1]{inputenc}
\usepackage{amsmath}
\usepackage{amsfonts}
\usepackage{amssymb}
\usepackage{mathrsfs}
\usepackage{booktabs}
\usepackage{url}
\usepackage[dvips]{graphicx}
\usepackage[skip=1pt,font=scriptsize]{subcaption}
\usepackage{ifthen}
\usepackage{xspace}
\usepackage{dsfont}
\usepackage{bbm}
\usepackage{cite}
\usepackage{epsfig}
\usepackage{epstopdf}
\usepackage{array}
\usepackage{multicol}
\usepackage{algpseudocode}
\usepackage{algorithm}
\usepackage{amssymb}
\usepackage{esint}
\usepackage[skip=2pt,font=footnotesize]{caption}
\usepackage{multicol}
\usepackage{amsthm}
\usepackage{setspace}
\usepackage{lipsum}
\usepackage{dblfloatfix}
\theoremstyle{plain}

\newtheorem{lemma}{Lemma}
\theoremstyle{definition}

\theoremstyle{remark}

\usepackage{footnote}

\usepackage{cancel}

\newcommand{\nm}[1]{}
\newcommand{\mb}[1]{\textcolor{blue}{\textbf{[MB: #1]}}}

\newcommand{\sst}[1]{}

\title{Power-Constrained Trajectory Optimization for Wireless UAV Relays with Random Requests
}
\author{Matthew Bliss and Nicol\`{o} Michelusi
\thanks{Bliss and Michelusi are with the School of Electrical and Computer Engineering, Purdue University, West Lafayette, IN, USA; emails: \{blissm,michelus\}@purdue.edu.}
}

\begin{document}
\maketitle
\thispagestyle{empty}
\pagestyle{empty}
\setulcolor{red}
\setul{red}{2pt}
\setstcolor{red}
\begin{abstract}
This paper studies the adaptive trajectory design of a rotary-wing UAV serving as a relay between ground nodes dispersed in a circular cell and a central base station.
Assuming the ground nodes generate uplink data transmissions randomly according to a Poisson process,
we seek to minimize the expected average communication delay to service the data transmission requests, subject to an average power constraint on the mobility of the UAV. The problem is cast as a semi-Markov decision process, and it is shown that the policy exhibits a two-scale structure, which can be efficiently optimized: in the outer decision, upon starting a communication phase, and given its current radius, the UAV selects a target end radius position so as to optimally balance a trade-off between \emph{average long-term} communication delay and power consumption; in the inner decision, the UAV selects its trajectory between the start radius and the selected end radius, so as to \emph{greedily} minimize the delay and energy consumption to serve the current request. Numerical evaluations show that, during waiting phases, the UAV circles at some optimal radius at the most energy efficient speed, until a new request is received. 
 Lastly, the expected average communication delay and power consumption of the optimal policy is compared to that of static and mobile heuristic schemes, demonstrating a reduction in latency by over 50\% and 20\%, respectively.
\end{abstract}


\vspace{-1.1mm}
\section{Introdcution}\label{section_introduction}
Much recent research has gone into studying UAVs operating in wireless networks \cite{SurveyUAVCell,UAVTutorial,UAVSurveyImpIss,FundTrade}. The primary motivation for this interest is due to the unique benefits that UAVs acting as flying base stations, mobile relays, etc., provide in improving the overall network performance over terrestrial infrastructure in terms of mobility, maneuverability, and enhanced line-of-sight (LoS) link probability \cite{SurveyUAVCell,UAVTutorial,UAVSurveyImpIss,FundTrade,OnlineTrajOptUS}. 

Already, the literature has shown that consideration of UAV deployment strategies,
in terms of optimal positioning or trajectory design, can go a long way to increase network performance of many of the useful metrics. In \cite{DynBSRepos}, dynamic repositioning led to increase in \emph{spectral efficiency} over heuristics under both FDMA and TDMA schemes. The works of \cite{MSCD,CoopRelaying} utilized UAVs in communication and maximized the \emph{total service time}, with \cite{MSCD} outperforming static and random deployment methods, and \cite{CoopRelaying} meeting BER requirements.

Although showing potential to improve these performance metrics, the design of UAV deployment strategies is not without challenges \cite{SurveyUAVCell,UAVTutorial,UAVSurveyImpIss}. Optimal trajectory design must be formulated appropriately to incorporate realistic constraints imposed on the UAVs. Already, works such as \cite{CoopRelaying,EnEffUAVComm,EnMin} have gone at length to incorporate constraints on the limited onboard energy and mission times inherent to low-altitude platforms (LAPs) \cite{SurveyUAVCell}.

Despite the enormous interest in the design of UAV-assisted wireless communication networks, most of the prior work focuses on \emph{deterministic} models, in which the data traffic generated by ground nodes (GNs) is known beforehand, and there are no uncertainties in the network dynamics. However, this is impractical in realistic systems, where uncertainty dominates, and random request arrivals must be accounted for. Therefore, the design of UAV-assisted communication networks with random data traffic is still an \emph{open problem}.

To address this problem, our previous paper \cite{OnlineTrajOptUS} considered the optimization of the UAV trajectory and communication strategy under random traffic generated by two GNs. However, that model was limited to two GNs, neglected the power consumption of the UAV, and assumed that the UAV is the destination of the data traffic.
In this paper, we extend the model to densely deployed GNs that need to transmit data payloads to a backbone-connected base station; we investigate the optimal trajectory and communication strategy of the UAV, under an average power constraint on the UAV mobility.

To investigate this added complexity, we consider a scenario in which an UAV acts as a relay between multiple GNs dispersed uniformly in a circular cell and a central base station (BS), receiving transmission requests from the GNs according to a Poisson process.
 We formulate the problem as that of designing an \emph{adaptive} trajectory, with the goal to minimize the average long-term communication delay incurred to serve the requests of the GNs, subject to a constraint on the long-term average UAV power consumption to support its mobility. 

We show that the optimal trajectory in the communication phase operates according to a two-scale decision-making process, which can be efficiently optimized: in the \emph{outer decision}, the UAV, given its current radius, selects a target end radius position, which balances optimally the trade-off between \emph{average long-term} communication delay and power; in the \emph{inner decision}, given its current radius and the selected end radius, the UAV \emph{greedily} minimizes the delay and energy trade-off to serve the current request. Our numerical results reveal that during \emph{waiting phases}, the UAV tends to circle at some optimal radius at an energy-efficient speed determined by the \emph{outer decision} process, until receiving a new uplink transmission request. Additionally, we show that the optimal trajectory design vastly outperforms sensible heuristics in terms of delay minimization: it outperforms a static hovering scheme by roughly 50\% and a mobile heuristic scheme by up to 20\%,
while maintaining the same average power consumption on the UAV.

The rest of the paper is organized as follows. In Sec. \ref{section_sys}, we introduce the system model, state the optimization problem, and cast it as a semi-Markov decision process; in Sec. \ref{sec:Decomp}, we present the two-scale optimization approach; in Sec. \ref{sec:Results}, we provide numerical results; lastly, in Sec. \ref{sec:Conclusions}, we conclude the paper with some final remarks.
\vspace{-2.0mm}
\section{System Model and Problem Formulation}\label{section_sys}
Consider the scenario depicted in Fig. \ref{fig:SysModel}, where
multiple ground nodes (GNs) distributed uniformly over a circular cell of radius $a$ randomly generate data packets of $L$ bits, that need to be transmitted to a backbone-connected base station (BS), located in the center of the cell in position $\mathbf{q}_B=(0,0)$.\footnote{Unless otherwise stated, we use polar coordinates to express positions, so that $\mathbf q=(r,\theta)$ denotes the position at distance $r$ from the center, with angle $\theta$ with respect to the $x$-axis, as shown in Fig. \ref{fig:SysModel}.}
 The density of GNs within the circular radius is denoted as $\lambda_G$ [GNs/m$^2$].
 Each GN generates uplink transmission requests according to a Poisson process with rate $\lambda_P$ [requests/GN/sec].
Overall, uplink transmission requests of $L$ bits arrive in time according to a Poisson process with rate $\lambda=\lambda_G{\cdot}\lambda_P$ [requests/sec/m$^2$]. Requests are received uniformly within the circular cell, so that the probability density function (pdf) of a request received in position $(r,\psi)$ is expressed as
\begin{equation}\label{eq:LocationPDF}	
f_{R,\Psi}(r,\psi) = 
\frac{r}{\pi a^2},\ \forall r\leq a,\psi \in [0,2\pi).
\end{equation}

Direct communication between the GNs and the BS might not be possible due to severe pathloss. An UAV is thus deployed, flying at a fixed height $H_U$, to relay the traffic between the GNs and the BS. Let $\mathbf{q}_{U}(t) = (r_U (t), \psi_U (t)) \in \mathbb R_{+} \times[0,2\pi)$ be
 the projection of its position on the ground surface at time $t$. Due to constraints on UAV mobility, its speed is subject to a maximum constraint, expressed in polar coordinates as
\begin{align}
\label{vconstraint}
	v_U(t)\triangleq \sqrt{(r'_U (t))^{2} {+} (r_U (t) \cdot \psi'_U (t))^{2}} \leq V_{\mathrm{max}},
\end{align}
where $f'$ denotes the derivative of $f$ with respect to time.

\begin{figure}[t]
\centering
\includegraphics[width=.88\linewidth]{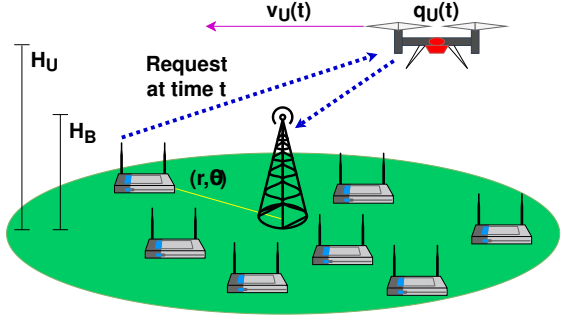}
\caption{System model depiction of an uplink transmission request at time $t$.}
\label{fig:SysModel}
\vspace{-5mm}
\end{figure}

We model the instantaneous UAV power consumption as
\begin{equation}
	P_{U}(t) = P_{\mathrm{c}}(t) + P_{\mathrm{mob}}\big( v_{U}(t) \big),
\end{equation}
where $P_{\mathrm{c}}(t)$ is the total power used onboard for communication processes and $P_{\mathrm{mob}}\big(v_{U}(t)\big)$ is the forward flight mobility power, a non-convex function of the UAV speed $v_U(t)$ \cite{EnMin,EnduranceEstimation}. We use the model in \cite{EnMin,EnduranceEstimation}, wherein
\begin{equation}\label{eq:PVModelCont}
	P_{\mathrm{mob}}(V){=}P_0 \Bigg(\!1{+}\frac{3V^2}{U_{\mathrm{tip}}^2}\!\Bigg){+}P_i \sqrt{ \sqrt{1{+}\frac{V^4}{4v_{0}^4}} - \frac{V^2}{2v_0^2} }{+} \beta V^3,
\end{equation}
where $P_0$ and $P_i$ are scaling constants, $U_{\mathrm{tip}}$ is the rotor blade tip speed, $v_0$ is the mean rotor induced velocity while hovering, $\beta \triangleq d_0 \rho s A/2$, $d_0$ is the fuselage drag ratio, $s$ is the rotor solidity, $\rho$ is the air density and $A$ is the rotor disc area (see \cite{EnMin}). An example using the physical parameters in  \cite{EnMin} is shown in Fig. \ref{fig:PVPlot}. Interestingly, the most power efficient operation is not achieved while hovering, but while flying at a fixed speed $V_{\min}{\simeq} 20$[m/s]; in addition, it is noted that the communication power $P_c(t)$ (order of 1W, as in \cite{EnMin}) is dwarfed by the amount of power used for UAV flight, $P_{\mathrm{mob}}(V)$ (order of x$100$W). Thus in this paper we will neglect $P_{\mathrm{c}}(t)$, and approximate $P_{U}(t) \approx P_{\mathrm{mob}} \big(v_{U}(t) \big)$.
\begin{figure}[t]
\centering
\includegraphics[width=.7\linewidth]{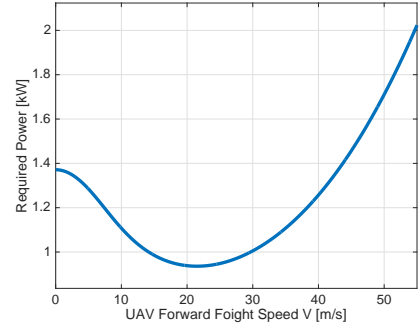}
\caption{UAV power vs. forward flight speed, simulated with the same physical parameters found in \cite{EnMin}.}
\label{fig:PVPlot}
\vspace{-6mm}
\end{figure}

We assume that communication intervals experience line of sight (LoS) links, and that the channel faces no probabilistic elements. In fact, UAVs in LAPs generally tend to have a highly likely occurrence of LoS links \cite{LOSDominance}. When an uplink transmission request is received at time $t$ from a GN in position $(r_G,\psi_G)$, the \emph{communication phase} begins, in which the UAV first receives the data payload from the GN, and then forwards it to the BS using a decode-and-forward, \emph{fly-hover-communicate} strategy \cite{EnMin}. Any additional requests received during this communication phase are dropped.\footnote{Alternatively, they might be served directly from the BS. This possibility will be investigated in our future work.} This phase is constituted of four distinct operations:
\begin{enumerate}[leftmargin=*]
\item The UAV flies from its current position $\mathbf q_U(t)=(r_U,\psi_U)$ to a new position
$\mathbf q_{GU}=(r_{GU},\psi_{GU})$, at constant speed $v_{1}$; the duration of this operation is
\begin{equation}
\Delta_{1}=v_{1}^{-1}\sqrt{r_{GU}^2 {+} r_U^2 {-} 2 r_{GU} \cdot r_U \cos{(\psi_{GU} {-} \psi_U)}},
\end{equation}
and its energy cost is $E_{1}=\Delta_{1}P_{\mathrm{mob}} \big(v_{1}\big);$
\item The UAV hovers in position $\mathbf q_{GU}$ while the GN transmits its data payload of $L$ bits to the UAV; assuming a fixed transmission power $P_{\mathrm{GN}}$,
the transmission rate is
\begin{equation}
	R_{GU}(d_{GU}) \triangleq B \log_{2} \Big( 1 + \frac{\gamma_{GU}}{d_{GU}^{2}}\Big),
\end{equation}
where $B$ is the channel bandwidth, $\gamma_{GU}$ is the SNR of the GN$\to$UAV link referenced at $1$-meter, and
\begin{equation}
	d_{GU}{=}\sqrt{H_U^2 {+} r_{GU}^2 {+} r_G^2 {-} 2 r_{GU} \cdot r_G \cos{(\psi_{GU} {-} \psi_G)}}
\end{equation}
is the UAV-GN distance; the associated duration and energy cost are
\begin{equation}
\Delta_2=\frac{L}{R_{GU}(d_{GU})},\ \ \ 
E_2=\Delta_2P_{\mathrm{mob}} \big(0\big);
\end{equation}
\item
The UAV flies from its current position $\mathbf q_{GU}$ to a new position
$\mathbf q_{UB}=(r_{UB},\psi_{UB})$, at constant speed $v_{3}$;
 the duration of this operation is
\begin{equation}
\nonumber
\Delta_3=v_3^{-1}\sqrt{r_{GU}^2 {+} r_{UB}^2 {-} 2 r_{GU} \cdot r_{UB} \cos{(\psi_{GU} {-} \psi_{UB})}},
\end{equation}
and its energy cost is $E_3=\Delta_3P_{\mathrm{mob}} \big(v_3\big);$
\item The UAV hovers in position $\mathbf q_{UB}$ while it relays the data payload to the BS; assuming a fixed transmission power $P_{\mathrm{UAV}}$ and reuse of the same frequency band,
the transmission rate is given by
\begin{equation}
	R_{UB}(d_{UB}) \triangleq B \log_{2} \Big( 1 + \frac{\gamma_{UB}}{d_{UB}^{2}}\Big),
\end{equation}
where $\gamma_{UB}$ is the SNR of the UAV$\to$BS link referenced at $1$-meter, $d_{UB}$ is  the UAV-BS distance,
\begin{equation}
	d_{UB} = \sqrt{(H_U-H_B)^2 {+} r_{UB}^2},
\end{equation}
and $H_B$ is the height of the BS antenna; the duration and energy cost of this operation are
$$
\Delta_4=\frac{L}{R_{UB}(d_{UB})},\ \ \ E_4=\Delta_4P_{\mathrm{mob}} \big(0\big).
$$
\end{enumerate}
The positions $\mathbf q_{GU},\mathbf q_{UB}$ and speeds 
$v_1,v_3$ are part of the design.
Overall, the delay and energy cost of the communication phase to serve a GN 
in position $(r_G,\psi_G)$ is given by
\begin{align}
&\!\!E^{(c)}{=}\Delta_1P_{\mathrm{mob}}\big(v_1\big){+}\Delta_3P_{\mathrm{mob}} \big(v_3\big)
{+}[\Delta_2{+}\Delta_4]P_{\mathrm{mob}} \big(0\big),\label{energy}\\
&\!\!\Delta^{(c)}=\Delta_1+\Delta_2+\Delta_3+\Delta_4.\label{delta}
\end{align}
Once the communication phase is completed at time $t+\Delta^{(c)}$, the UAV, now in position $\mathbf q_U(t+\Delta^{(c)}){=}\mathbf q_{UB}$, enters the \emph{waiting phase}, where it awaits for new  requests, and the process is repeated indefinitely. During this process of communication and waiting for new requests, the UAV follows a trajectory, part of our design, with the goal to minimize the average  communication delay, subject to an average UAV power constraint (which, given an onboard battery capacity, translates into an endurance constraint), as formulated next.

\subsection{Problem Formulation}
Let $\Delta_u^{(c)}$ and $E_u^{(c)}$ be the delay and energy cost incurred to complete the \emph{communication phase} of the $u$th request serviced by the UAV, as given by \eqref{energy} and \eqref{delta}. Let $\Delta_u^{(w)}$ and $E_u^{(w)}$ be the duration and energy cost of the waiting phase preceding it. Let $T_u{=}\Delta_u^{(w)}{+}\Delta_u^{(c)}$ and $E_u{=}E_u^{(w)}{+}E_u^{(c)}$ be the duration and energy cost of the $u$th waiting and communication cycle. Let $M_t$ be the total number of requests served and completed up to time $t$.
Then, we define the expected average communication delay and average UAV power consumption under a given trajectory policy $\mu$ (defined later), with the UAV starting from the geometric center $\mathbf{q}_{U}(0) = (0,0)$ as
\begin{align}\label{eq:OverallObj}
& \bar{D}_{\mu} \triangleq \lim_{t \rightarrow \infty} \mathbb{E}_{\mu} \left[\frac{\sum_{u = 0}^{M_t - 1} \Delta_u^{(c)}}{M_t}\Bigg|\mathbf{q}_{U}(0) = (0,0) \right],
\end{align}
\vspace{.1mm}
\begin{align}
& \bar{P}_{\mu} \triangleq \lim_{t \rightarrow \infty} \mathbb{E}_{\mu} \left[ \frac{\sum_{u = 0}^{M_t - 1}E_u}{\sum_{u = 0}^{M_t - 1}T_u} \Bigg| \mathbf{q}_{U}(0) = (0,0) \right] \label{eq:AvgPDef}.
\end{align}
We then seek to solve
\begin{align}\label{eq:OverallObj}
	\bar{D}_{\mu}^{*} = \underset{\mu}{\mathrm{min}} \; \bar{D}_{\mu} 
	\ \mathrm{s.t.} \; \bar{P}_{\mu}  \leq P_{\mathrm{avg}},
\end{align}
whose minimizer is denoted by the optimal policy $\mu^*$.
\\
\indent
Note that this problem is non-trivial. Consider, for instance, the power unconstrained delay minimization problem in which the UAV is the only destination of the packets. Here, the minimum delay to serve a request is achieved by flying towards the GN at maximum speed to improve the link quality. However, this strategy may not be optimal in an \emph{average} delay sense, due to the potentially longer distance that must be covered by the UAV to serve a subsequent request. Thus, it might be preferable for the UAV to operate closer to the geometric center of the cell, where new requests can more readily be served,
as observed in our earlier work in \cite{OnlineTrajOptUS}. Intriguingly, under an average power constraint, more interesting tradeoffs may emerge. For instance, maximizing the speed to improve the link quality may no longer be a viable option, due to high power consumption, and hovering might not be the most energy efficient operation (see Fig. \ref{fig:PVPlot}).
Letting
\begin{align*}
	&[\bar{E}_{\mu},\bar{T}_{\mu} ] \triangleq \lim_{t \rightarrow \infty} \mathbb{E}_{\mu} \Bigg[ \frac{\sum_{u = 0}^{M_t - 1}[E_u,T_u]}{M_t} \Bigg| \mathbf{q}_{U}(0) = (0,0) \Bigg],
\end{align*}
be the average energy and time of a waiting and communication cycle, we can use Little's Theorem (see \cite{LittlesTheorem}) to express the average power as $\bar{P}_{\mu}=\bar{E}_{\mu}/\bar{T}_{\mu}$, so that the power constraint can be equivalently expressed as $\bar{E}_{\mu}-P_{\mathrm{avg}}\bar{T}_{\mu} \leq 0$. Hence, \eqref{eq:OverallObj} can also be expressed equivalently as
\begin{align}\label{eq:OverallPolRef}
	\mu^* = \mathrm{arg}\underset{\mu}{\mathrm{min}} \; \bar{D}_{\mu}\ 
	\mathrm{s.t.} \; \bar{E}_{\mu} - P_{\mathrm{avg}} \bar{T}_{\mu} \leq 0.
\end{align}
As we will see, this is a more tractable form to work with, because it removes the metric $\bar{T}_{\mu}$ from the denominator, and allows one to express the optimization problem as a semi-Markov decision processes (SMDP).
\\
\indent To deal with the inequality constraint in \eqref{eq:OverallPolRef}, we let
\begin{align}\label{eq:DualFunction}
&	g(\nu) = \underset{\mu}{\mathrm{min}} \; \bar{D}_{\mu} + \nu (\bar{E}_{\mu} - P_{\mathrm{avg}}\bar{T}_{\mu})
\\&	=\underset{\mu}{\mathrm{min}}
	\lim_{t \rightarrow \infty} \mathbb{E}_{\mu} \Bigg[\frac{\sum_{u = 0}^{M_t - 1} (\Delta_u^{(c)}{+}\nu E_u{-}\nu P_{\mathrm{avg}}T_u)}{M_t}\Bigg|\mathbf{q}_{U}(0) \Bigg],
	\nonumber
\end{align}
be the dual function with dual variable $\nu$, which can be optimized by solving the related dual maximization problem
\begin{equation}\label{eq:DualProblem}
	\underset{\nu \geq 0}{\mathrm{max}} \; g(\nu).
\end{equation}

\subsection{SMDP Formulation}
In this section, we formulate \eqref{eq:DualFunction} as a SMDP, and characterize its states, actions, cost metrics, and policy. We then optimize this SMDP via discretization and dynamic programming. In general, the state at any given time $t$ requires knowledge of the UAV position $\mathbf q_U(t){=}(r_U(t),\psi_U(t))$, whether there is a request for uplink transmission, and if a request exists, the location of the GN that originated it, $\mathbf q_G(t){=}(r_G(t),\psi_G(t))$. However, during waiting phases, the angular coordinate $\psi_U(t)$ of the UAV is irrelevant to the decision process; only its radius $r_U(t)$ is. In fact, the angular coordinate of requests is uniform in $[0,2\pi)$, irrespective of $\psi_U(t)$. During communication phases, only the position of the GN relative to that of the UAV matters, i.e., their radii $r_U(t)$, $r_G(t)$ and relative angular coordinate
$\theta_G(t)\triangleq\psi_G(t){-}\psi_U(t){\in}[0,2\pi)$,
 which is uniformly distributed in $[0,2\pi)$. Hence, the state can be more compactly expressed as $r_U(t)$ (for the UAV position) and 
 $(r_G(t),\theta_G(t)\triangleq\psi_G(t)-\psi_U(t))$ (for a request). Let
\begin{equation}
	\mathcal{R}_{\mathrm{UAV}} \triangleq \mathbb R_+,\ \mathcal{Q}_{\mathrm{GN}} \triangleq [0,a] \times [0,2\pi)
\end{equation}
be the set of all radii positions of the UAV, $r_U\in\mathcal{R}_{\mathrm{UAV}}$, and of all polar coordinates of requests by the GNs, relative to the angular coordinate of the UAV, $(r_G,\theta_G)\in\mathcal{Q}_{\mathrm{GN}}$. With that,
we define the set of \emph{waiting} and \emph{communication states} as
$$\mathcal{S}_{\mathrm{wait}} = \mathcal{R}_{\mathrm{UAV}} \times \{(-1,-1)\},\ 
\mathcal{S}_{\mathrm{comm}} = \mathcal{R}_{\mathrm{UAV}} \times \mathcal{Q}_{\mathrm{GN}},$$
where $({-}1,{-}1)$ denotes no active request,
so that the overall state space is $\mathcal{S}{=}\mathcal{S}_{\mathrm{wait}}{\cup}\mathcal{S}_{\mathrm{comm}}$.
To define this SMDP, we sample the continuous time interval to define a sequence of states $\{s_n, n{\geq}0 \} {\subseteq}\mathcal{S}$
 with the Markov property,
 along with associated time, delay and energy costs,
  as specified below.

If the UAV is in state $s_n = (r_U,-1,-1) \in \mathcal{S}_{\mathrm{wait}}$ at time $t$, i.e., it is in the radial position $r_U$ and there are no active requests, then the actions available are to move with velocity vector $(v_r,\theta_c)$, over an arbitrarily small but fixed interval of duration $\Delta_{0} \ll \pi a^2 \lambda$. The terms $v_r$ and $\theta_c$ denote the radial and angular velocity components, respectively;
since they must obey the velocity constraint \eqref{vconstraint},
they take values from the action space
\begin{equation}\label{eq:WaitActions}
\mathcal{A}_{\mathrm{wait}}(r_U) {\triangleq}\Big\{(v_{r},\theta_{c}) \in \mathbb R^{2} \, \Big| \, \sqrt{v_{r}^{2} + r_U^2\cdot \theta_{c}^{2}}\leq V_{\mathrm{max}} \Big\}.
\end{equation}
After this interval, the new radial position becomes $r_U(t+\Delta_0) = r_U+v_r \Delta_0$.
Moreover, with probability 
$e^{-\pi a^2 \lambda \Delta_{0}}$, no new request has been received in the time interval $[t,t+\Delta_0]$ so that the new state $s_{n+1}$, sampled at time $t+\Delta_0$, becomes
$s_{n+1}=(r_U+v_r\Delta_0,-1,-1)\in\mathcal{S}_{\mathrm{wait}}$.
Otherwise, a new request is received in position $(r_G,\theta_G)$ with the pdf given by \eqref{eq:LocationPDF},
so that the new state is $s_{n+1}{=}(r_U+v_r\Delta_0,r_G,\theta_G)\in\mathcal{S}_{\mathrm{comm}}$.

Overall, the transition probability from the waiting state $s_n = (r_{U},-1,-1)$ under action $\mathbf{a}_n{=}(v_r,\theta_c)$ is expressed as
\begin{align*}
&\mathbb{P}(s_{n+1}{=}(r_U{+}v_r\Delta_0,{-}1,{-}1)|s_n,\mathbf{a}_n) = e^{-\pi a^2 \lambda \Delta_{0}}, 
\\
&\mathbb{P}(s_{n+1} {\in} (r_U + v_r \Delta_0,\mathcal{F}) \,|s_n,\mathbf{a}_n) 
= \frac{A(\mathcal{F}) {\cdot} (1{-}e^{-\pi a^2 \lambda \Delta_{0}})}{\pi a^2},
\end{align*}
$\forall\mathcal{F}{\subseteq}\mathcal{Q}_{\mathrm{GN}}$, where $A(\mathcal{F})$ is the area of the region $\mathcal{F}$ on the $x$-$y$ plane. To complete the definition of the SMDP, we need to define the cost metrics under each state and action.
The duration of  action $\mathbf a_n{=}(v_r,\theta_c)$ in state $s_n{=}(r_U,-1,-1)$ is $T(s_n,\mathbf a_n){\triangleq} \Delta_0$, its delay cost is $\Delta(s_n,\mathbf a_n){\triangleq}0$ (the UAV is not communicating), and its energy cost is $E(s_n,\mathbf a_n){\triangleq}\Delta_0 P_{\mathrm{mob}}\Big(\sqrt{v_r^2{+}r_U^2{\cdot}\theta_c^2} \Big)$.

Upon reaching state $s_n{=}(r_U,r_G,\theta_G){\in}\mathcal{S}_{\mathrm{comm}}$ at time $t$, the UAV has received a request to serve the transmission of $L$ bits from a GN located at $(r_G,\theta_G)$ (relative to its current angle coordinate). The actions available to the UAV at this point are all trajectories starting from its current position that follow the $4$-step communication phase procedure described earlier. Thus, we denote an action as $\mathbf a_n = (\mathbf q_{GU},v_1,\mathbf q_{UB}, v_3)$, whose duration and communication delay is $T(s_n,\mathbf a_n)=\Delta(s_n,\mathbf a_n)\triangleq\Delta^{(c)}$ as given by \eqref{delta}, and whose energy cost is $E(s_n,\mathbf a_n)\triangleq E^{(c)}$ as given by \eqref{energy}. After the communication phase is completed at time $t+\Delta^{(c)}$, the new state $s_{n+1}$ is sampled. At this point, a new waiting phase begins and the radial position of the UAV is $r_{UB}$ (the radial position corresponding to $\mathbf q_{UB}$), so that the transition probability from state $s_n$ under action $\mathbf a_n$ is expressed as
\begin{align}
	\mathbb{P}(s_{n+1} = (r_{UB},-1,-1)\,|\,s_n=(r_U,r_G,\theta_G),\mathbf a_n) = 1. \label{eq:CommTrx}
\end{align}

With the states and actions defined, we can define a policy $\mu$. Specifically, for states $(r_U,-1,-1) \in \mathcal{S}_{\mathrm{wait}}$, $\mu$ selects a velocity vector $(v_r,\theta_c)\in\mathcal{A}_{\mathrm{wait}}(r_U)$, as defined in \eqref{eq:WaitActions}. Likewise, for states $(r_U,r_G,\theta_G) \in \mathcal{S}_{\mathrm{comm}}$, the policy selects an action $(\mathbf q_{GU},v_1,\mathbf q_{UB}, v_3)$ as has been prescribed.

Having now defined a \emph{stationary policy} $\mu$, we can reformulate the Lagrangian term $L_\mu^{(\nu)}\triangleq\bar{D}_{\mu} + \nu (\bar{E}_{\mu} - P_{\mathrm{avg}}\bar{T}_{\mu})$, which we seek to minimize to find the dual function in \eqref{eq:DualFunction}. In the context of the SMDP,
this can be expressed as
\begin{equation}\label{eq:OverallObjH}
L_\mu^{(\nu)}= \lim_{K \rightarrow \infty} \mathbb{E} \Bigg[ \frac{\frac{1}{K}\sum_{n=0}^{K-1}  
\ell_\nu(s_n,\mu(s_n))
 }{\frac{1}{K}\sum_{n = 0}^{K-1} \chi(s_n \in \mathcal{S}_{\mathrm{comm}})} \Bigg| s_0 \Bigg],
\end{equation}
where $s_0 = (0,-1,-1)$ (the UAV begins in the center with no transmission requests), $\chi(C)$ is the indicator function of the event $C$, and we have defined the overall Lagrangian metric in state $s$ under action $\mathbf a$ as
\begin{equation}\label{eq:EllLagDef}
\ell_\nu(s,\mathbf a)\triangleq\Delta(s,\mathbf a) + \nu \big(E(s,\mathbf a) - P_{\mathrm{avg}} T(s,\mathbf a) \big).
\end{equation}
Using Little's Theorem \cite{LittlesTheorem}, we can rewrite 
$L_\mu^{(\nu)}$ in terms of the steady-state pdf of being in state $s$ in the SMDP, $\Pi_{\mu}(s)$,
\begin{equation}\label{eq:CostMetric}
L_\mu^{(\nu)}= 
 \frac{1}{\pi_{\mathrm{comm}}}\int_{\mathcal{S}} \Pi_{\mu}(s)\ell_\nu(s,\mu(s))\mathrm ds,
\end{equation}
where $\pi_{\mathrm{comm}}\triangleq\int_{\mathcal{S}_{\mathrm{comm}}} \Pi_{\mu}(s) \mathrm ds$ 
is the steady-state probability of being in a communication state in the SMDP,
which is provided in closed-form in the next lemma.

\begin{lemma}\label{eq:SSComm}
	Let $\pi_{\mathrm{wait}}$ and $\pi_{\mathrm{comm}}$ be the steady-state probabilities that the UAV is in the waiting and communication phases in the SMDP. We have that
\begin{equation}\label{eq:SSCommResult}
	\pi_{\mathrm{wait}} = \frac{1}{2-e^{-\pi a^2 \lambda \Delta_0}},\  \; \pi_{\mathrm{comm}} = \frac{1-e^{-\pi a^2 \lambda \Delta_0}}{2-e^{-\pi a^2 \lambda \Delta_0}}.
\end{equation}
\end{lemma}
\begin{proof}
Let $p_{ww}$, $p_{wc}$, $p_{cw}$, and $p_{cc}$ 
be the probabilities of remaining in the waiting phase ($ww$), moving from a waiting state to a communication state ($wc$),
from a communication to a waiting state ($cw$), or remaining in the communication phase ($cc$), in one state transition of the SMDP. Then, $p_{ww} = e^{-\pi a^2 \lambda \Delta_0}$ (if no request is received, the SMDP remains in the waiting state), $p_{wc}=1-p_{ww}$, $p_{cw} = 1$, and $p_{cc} = 0$ (after the communication phase, the waiting phase begins, see \eqref{eq:CommTrx}). Therefore, $\pi_{\mathrm{wait}}$ and $\pi_{\mathrm{comm}}$ satisfy $\pi_{\mathrm{wait}} + \pi_{\mathrm{comm}} = 1$ and
\begin{align}
&		\pi_{\mathrm{wait}} = p_{ww}\pi_{\mathrm{wait}} + p_{cw}\pi_{\mathrm{comm}} \nonumber
		=e^{-\pi a^2 \lambda \Delta_0}\pi_{\mathrm{wait}} + \pi_{\mathrm{comm}},
		\\&
		\pi_{\mathrm{comm}} = p_{wc}\pi_{\mathrm{wait}} + p_{cc}\pi_{\mathrm{comm}}
		=(1-e^{-\pi a^2 \lambda \Delta_0})\pi_{\mathrm{wait}}, \nonumber
\end{align}
whose solution is given in the statement of the lemma.
\end{proof}
The minimization problem of \eqref{eq:DualFunction} can then be expressed as
\begin{equation}\label{eq:TotalGMin}
	g(\nu)=\underset{\mu}{\mathrm{min}}\ L_\mu^{(\nu)} = \frac{1}{\pi_{\mathrm{comm}}}\underset{\mu}{\mathrm{min}} \; \int_{\mathcal{S}} \Pi_{\mu}(s) 
	\ell_\nu(s,\mu(s))\mathrm d s\!
\end{equation}
with the subsequent dual maximization problem given in \eqref{eq:DualProblem}.
\vspace{-5mm}
\section{Two-scale SMDP Optimization} \label{sec:Decomp}
To reduce the complexity of the problem, we exploit a decomposition of the policy $\mu$, such that the total optimization problem of \eqref{eq:TotalGMin} and its dual maximization consist of solving simpler, inner and outer optimization problems separately in a two-scale decision-making approach.

Note that the steady-state pdf $\Pi_{\mu}(s)$ depends on the policy $\mu$ only through $v_r$ for \emph{waiting state} actions $(v_r,\theta_c)$ and through the radius $r_{UB}$ of $\mathbf q_{UB} = (r_{UB},\theta_{UB})$ for \emph{communication state} actions $(\mathbf q_{GU}, v_1, \mathbf q_{UB}, v_3)$. By separating $v_r$ from $\theta_c$ in the \emph{waiting states} and $r_{UB}$ from the other action elements in the \emph{communication states}, we have created a decomposition which allows parts of the optimal cost of a state-action pair, $\ell_{\nu}^* (s, \mu(s))$, to be solved separately from the steady-state probabilities $\Pi_{\mu}(s)$. Next, we formalize this decomposition.

Let $W(s) \triangleq v_r \in [-V_{\mathrm{max}},V_{\mathrm{max}}]$ define the \emph{radial velocity policy} of the \emph{waiting states} that specifies the radial velocity component of a waiting action $\mathbf{a}{=}(v_r,\theta_c){\in} \mathcal{A}_{\mathrm{wait}}(r_U)$. Additionally, let $U(s){\triangleq}r_{UB} {\in} \mathcal{R}_{\mathrm{UAV}}$ define the \emph{next radius position policy} of the \emph{communication states} that specifies the end radius position of the communication action. Under this decomposition, we can express the pdf $\Pi_{\mu}$ solely as a function of the policies $W,U$, i.e. $\Pi_{W,U}$, so that we have that
\begin{equation}\label{eq:PolDecomp}
	g(\nu) = \frac{1}{\pi_{\mathrm{comm}}}\,\underset{W,U}{\mathrm{min}} \int_{\mathcal{S}} \Pi_{W,U}(s) \ell_{\nu}^{*}(s,Z(s))\mathrm{d}s ,
\end{equation}
where $Z(s) = W(s)$ for $s \in \mathcal{S}_{\mathrm{wait}}$, $Z(s) = U(s)$ for $s \in \mathcal{S}_{\mathrm{comm}}$, and $ \ell_{\nu}^{*}(s,Z(s))$ is obtained by greedily minimizing 
$ \ell_{\nu}^{*}(s,\mathbf a)$ with respect to the components of the actions not specified by $W,U$ (hence not affecting $\Pi_{W,U}$). Namely, for waiting states $s = (r_U,-1,-1)$ and $W(s) = v_r$,
\begin{align}\label{eq:MinLWait}
&\ell_{\nu}^{*}(s,v_r) {=} \underset{\theta_c}{\mathrm{min}}\ \nu\left[ P_{\mathrm{mob}}\!\!\left(\sqrt{v_{r}^2 {+} r_U^2 {\cdot} \theta_c^2}\right) {-} P_{\mathrm{avg}} \right]\Delta_0, \nonumber \\
&\mathrm{s.t.}\;\; \sqrt{v_{r}^{2} {+} r_U^2 {\cdot} \theta_c^2} \leq V_{\mathrm{max}}.
\end{align}
Note that the minimizer $\theta_c^{*}$ of \eqref{eq:MinLWait} is the angular velocity that minimizes the UAV power consumption for some given radial velocity $v_r$ and UAV radius $r_U$, solvable \emph{offline} through exhaustive search (for the case in which the power curve follows the unimodal function given in Fig. \ref{fig:PVPlot}, it can be determined in closed form as
$\theta_c=0$ if $v_r\geq v_{\min}\triangleq\arg\min_{V>0}P_{\mathrm{mob}}(V)$
and $\theta_c=\sqrt{v_{\min}^2-v_{r}^{2}}/r_U$ if $v_r< v_{\min}$). For 
communication states
$s = (r_U, r_G, \theta_G)$ and $U(s) = r_{UB}$,
\begin{align}
&\ell_{\nu}^{*}(s,r_{UB}) = \underset{\underset{ v_1,v_3\leq V_{\mathrm{max}}}{\mathbf q_{GU},\theta_{UB}}}{\mathrm{min}}\; (1-\nu P_{\mathrm{avg}})\Delta^{(c)}+ \nu E^{(c)}
\end{align}
where $\Delta^{(c)}$ and $E^{(c)}$ are the delay and energy costs given by \eqref{energy} and \eqref{delta}. Due to low dimensionality, the solutions to $\ell_{\nu}^{*}(s,U(s))$ can also be found through an exhaustive search. 
Once $\ell_{\nu}^{*}(s,Z(s))$ has been determined for all states $s$, the problem of \eqref{eq:PolDecomp} can then be solved for a given value of $\nu$ through discretization and dynamic programming (i.e., value iteration or policy iteration), where the subsequent dual maximization resorts to adjusting $\nu$ iteratively, either by\sst{ a simple}\nm{dont use "simple"!} exhaustive search or subgradient-based methods.

\vspace{-1mm}
\section{Numerical Results} \label{sec:Results}
For the simulation parameters, we use a channel bandwidth $B {=} 1$MHz, 1-meter reference SNRs $\gamma_{GU} {=}\gamma_{UB} {=} 40$dB, UAV height $H_U {=} 120$m, BS height $H_B {=} 60$m, maximum UAV speed $V_{\mathrm{max}} {=} 55$m/s, and a Poisson arrival rate of $\lambda {=} 2.693 {\times} 10^{-9}$[requests/sec/m$^2$]. For the power consumption model, we use the power-speed relationship given in \eqref{eq:PVModelCont} (see Fig. \ref{fig:PVPlot}) and the same parameters utilized in \cite{EnMin}.

To solve an approximation of the problem, we discretize the state and action spaces, solving an inner SMDP for a given $\nu^{(k)}$ via value iteration, updating $\nu^{(k+1)}$, and repeating until the maximum is found.
In discretizing the state space, we select a cell radius of $a {=} 1600$m and $N {=} 10$ equispaced discretized radii, with a single GN located in the center, $M {=} 3$ equispaced angular positions at the next discretized radius position, $2M$ in the next one, $3M$ in the next one, and so on until we reach the $N$th radius value, ensuring that the distribution of GNs in the circular area approximates the pdf of the GNs described in the system model. We discretize the \emph{radial velocity} actions into $K {=} 13$ equispaced values such that $v_r {\in} \{-V_{\mathrm{max}},{\dots},0,{\dots},V_{\mathrm{max}}\}$. For the \emph{next radius position} actions, we utilize the radii positions indexed by the set $\{1,{\dots},N\}$. Lastly, the action duration $\Delta_0$ from waiting states is chosen in such a way that it is unlikely to receive more than one request in an interval $[0,\Delta_0]$. We choose $e^{-\pi a^2 \lambda \Delta_0} {\simeq} 0.93$.

In Fig. \ref{fig:WaitPol}, we depict the \emph{waiting phase} policy, where the target average power constraint is fixed at $P_{\mathrm{avg}} {=} 1000$ Watts, and the data payload value is varied for comparison. Note that for small data
payload values, more consideration is placed on power minimization, hence the UAV seeks to move towards a positive radius value and move circularly at the power-minimizing speed until receiving a request; for larger data payloads, the minimization focuses more on communication delay, hence, the UAV seeks to reach the geometric center of the GNs quickly in order to position itself best in anticipation of future transmission requests.

Next, we look at an example of the \emph{communication phase} policy for a fixed data payload value of $L = 1$Mbit, UAV position $(r_U,\theta_U) {=} (710,0)$m, request radius $r_G {=} 1600$m, and varied relative request angles, $\theta_G$, as illustrated in Fig. \ref{fig:CommPol}. We observe the pattern of the relative distance between a requesting GN and the associated position $\mathbf q_{GU}$, where the UAV first receives $L$ bits from the GN. Additionally, a common end radius of $r_{UB} {=} 178$m is found to be an optimal end point, even under variations in the relative request angles $\theta_G$.
\begin{figure}[t]
\centering
\includegraphics[width=.8\linewidth]{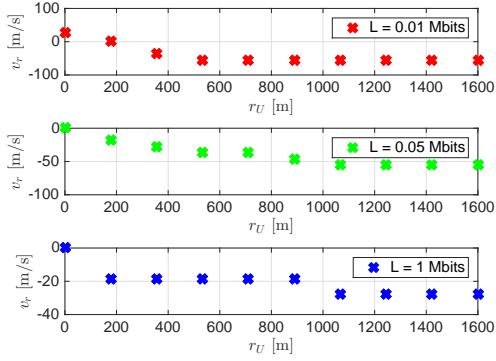}
\caption{Optimal radial velocities $v_r$ for each of the waiting state UAV radii, shown across $3$ different data payloads}
\vspace{-5mm}
\label{fig:WaitPol}
\end{figure}
\begin{figure}[t]
\centering
\includegraphics[width=.8\linewidth]{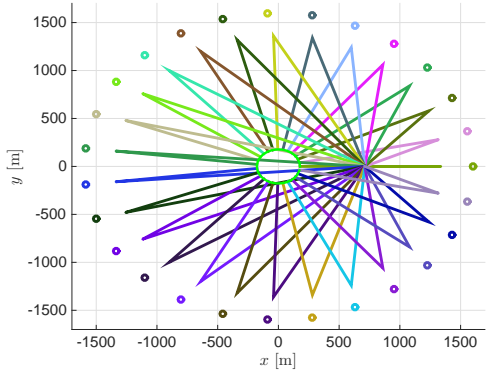}
\caption{Communication phase policy ($U^{*}(s) = r_{UB} = 178$ m and $v_1^*,v_3^* = 29.1$ m/s in all cases) with data payload $L = 1$ Mbits for UAV position $(r_U,\theta_U) = (710,0)$ m, $r_G = 1600$ m, and varied $\theta_G$.}
\label{fig:CommPol}
\vspace{-5mm}
\end{figure}

Finally, in Fig. \ref{fig:PwrVsDelayl}, we fix the data payload value at $L {=} 1$Mbit and show how the optimal expected average delay, $\bar{D}_{\mu}^{*}$, changes for various target $P_{\mathrm{avg}}$ values in the range $[875,1850]$Watts. As expected,  $\bar{D}_{\mu}^{*}$ decreases with increasing UAV power consumption. Additionally, the performance is compared to several other heuristics:
\begin{enumerate}[leftmargin=*]
\item \emph{Hover at center}: The UAV always hovers at the center of the cell. The expected
 average delay is $90.59$[s], noticeably worse than the delay yielded by the optimal policy $\mu^*$ for any of the tested $P_{\mathrm{avg}}$ targets. Alternatively, if the GNs always transmit directly to the BS, we found that the performance is roughly the same, since the UAV-BS link incurs small delay thanks to the small UAV-BS distance.\footnote{For this case, we use optimistically the same pathloss model as for the GN-UAV-BS links.}
\item \emph{Start-end at center}: The UAV hovers at the center, awaiting requests; once a request is received, it moves at speed $v$ towards the GN at a certain radius $r_{GU}$, receives the $L$ bits from the GN, travels back at speed $v$ to the center, where it hovers to relay the data payload to the BS; $r_{GU}$ is optimized to minimize the communication delay of each request, whereas $v$ is varied so as to obtain different power consumptions. Note that the expected average communication delay is worse than the optimal policy $\mu^*$ across all values of $P_{\mathrm{avg}}$,
except for smaller values of $P_{\mathrm{avg}}$ (due to the approximation introduced by discretization).
\end{enumerate}
Overall, the optimal policy outperforms these heuristic schemes by a significant margin, hence demonstrating the importance of an adaptive design that optimizes the average long-term performance.\sst{Thus in both the hover at center and start-end at center heuristic schemes, the UAV by and large outperforms
 in terms of expected average delay,}\nm{what did you mean? From what I read, you are saying that the heuristic schemes outperform the optimal scheme. Please restate. I reformulated it}\sst{ and it is sensible to implement the UAV trajectory over a wide range of target power constraints.}\nm{what do you mean?}
\begin{figure}[t]
\centering
\includegraphics[width=.8\linewidth]{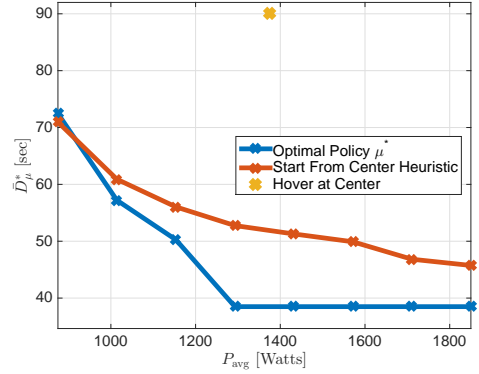}
\caption{Expected average communication delay vs. average power constraint for the optimal policy $\mu^*$ and two heuristic schemes, for data payload value $L = 1$ Mbit.}
\label{fig:PwrVsDelayl}
\vspace{-5mm}
\end{figure}
\section{Conclusions} \label{sec:Conclusions}
In this paper, we studied the trajectory optimization problem of one UAV acting as a relay in a two-phase, decode-and-forward strategy, servicing random uplink transmission requests by GNs seeking to communicate a data payload to a centralized BS. We formulated a continuous state and action space, discretized the problem into an SMDP, and solved it through dynamic programming. It was shown that the problem exhibits an interesting two-scale structure. Numerical evaluations demonstrate consistent improvements in the delay performance over sensible heuristics.\sst{, for a range of target average power constraints.}

\bibliographystyle{IEEEtran}
\bibliography{IEEEabrv,ref} 

\end{document}